\documentclass[11pt, letterpaper]{article}
\usepackage{fullpage}
\usepackage{amsthm}
\usepackage{amsmath,amssymb,amsfonts,nicefrac}
\usepackage{xspace}
\usepackage{color}
\usepackage{url}
\usepackage{hyperref}
\usepackage{bm}
\usepackage{bbm}
\usepackage{times}

\usepackage{enumitem}

\newtheorem{thm}{Theorem}[section]

\newtheorem{lemma}[thm]{Lemma}
\newtheorem{corollary}[thm]{Corollary}

\newtheorem{definition}[thm]{Definition}

\newcommand\card[1]{\left| {#1} \right|}
\newcommand\sett[2]{\left\{ \left. #1 \;\right\vert #2 \right\}}

\newcommand\set[1]{{\left\{ #1 \right\}}}
\newcommand\Prob[2]{{\Pr_{#1}\left[ {#2} \right]}}

\newcommand\Expect[2]{{\mathop{\mathbb{E}}_{#1}\left[ {#2} \right]}}

\newcommand\skipi{{\vskip 10pt}}

\newcommand\eps{\varepsilon}

\newcommand*\xor{\mathbin{\oplus}}

\renewcommand\geq{\geqslant}
\renewcommand\leq{\leqslant}

\newcommand{\rom}[1]{\uppercase\expandafter{\romannumeral #1\relax}}

\title{Parallel Repetition for the {\sf GHZ} Game: Exponential Decay}
\author{
Mark Braverman
\thanks{Department of Computer Science, Princeton University.
Research supported in part by the NSF Alan T. Waterman Award, Grant No. 1933331, a Packard Fellowship in Science and Engineering, and the Simons Collaboration on Algorithms and Geometry. }\and
Subhash Khot\thanks{Courant institute of Mathematical Sciences, New York University. Supported by
		the NSF Award CCF-1422159, NSF Award CCF-2130816, and the Simons Investigator Award.} \and
Dor Minzer\thanks{Department of Mathematics, Massachusetts Institute of Technology.  Supported by a Sloan Research Fellowship.}
}
\date{\vspace{-5ex}}
\begin{document}
\maketitle
\begin{abstract}
  We show that the value of the $n$-fold repeated {\sf GHZ} game is at most $2^{-\Omega(n)}$,
  improving upon the polynomial bound established by Holmgren and Raz. Our result is established
  via a reduction to approximate subgroup type questions from additive combinatorics.
\end{abstract}
\section{Introduction}
\subsection{Multi-player Parallel Repetition and the {\sf GHZ} Game}
The {\sf GHZ} game is a $3$-player game in which a verifier samples a triplet $(x,y,z)$ uniformly from $S = \sett{(x,y,z)}{x,y,z\in\{0,1\}, x\xor y\xor z = 0\pmod{2}}$,
then sends $x$ to Alice, $y$ to Bob and $z$ to Charlie. The verifier receives from each one of them a bit, $a$ from Alice, $b$ from Bob and $c$ from Charlie,
and accepts if and only if $a\xor b\xor c = x\lor y\lor z$. It is easy to prove that the value of the {\sf GHZ} game, ${\sf val}({\sf GHZ})$, defined as
the maximum acceptance probability of the verifier
over all strategies of the players, is $3/4$. The $n$-fold repeated {\sf GHZ} game is the game in which the verifier samples $(x_i,y_i,z_i)$ independently from
$S$ for $i=1,\ldots,n$, sends $\vec{x} = (x_1,\ldots,x_n)$, $\vec{y} = (y_1,\ldots,y_n)$ and $\vec{z} = (z_1,\ldots,z_n)$ to Alice, Bob and Charlie respectively,
receives vector answers $f(\vec{x}) = (f_1(\vec{x}),\ldots,f_n(\vec{x}))$, $g(\vec{y}) = (g_1(\vec{y}),\ldots,g_n(\vec{y}))$ and $h(\vec{z}) = (h_1(\vec{z}),\ldots,h_n(\vec{z}))$ and accepts
if and only if $f_i(\vec{x})\xor g_i(\vec{y})\xor h_i(\vec{z}) = x_i\lor y_i\lor z_i$ for all $i=1,\ldots,n$. What can one say about the value of the $n$-fold repeated game,
${\sf val}({\sf GHZ}^{\otimes n})$? As for lower bounds, it is clearly that case that ${\sf val}({\sf GHZ}^{\otimes n})\geq (3/4)^n$ and one expects that value
of the game to be exponentially decaying with $n$. Proving such upper bounds though is significantly
more challenging.

The {\sf GHZ} game is a prime example of a $3$-player game for which parallel repetition is not well understood. For $2$-player games, parallel
repetition theorems with an exponential decay have been known for a long time~\cite{Raz,Holenstein,Rao,BravermanGarg,DinurSteurer}, and in fact the state of the art parallel repetition theorems
for $2$-player games are essentially optimal. As for multi-player games, Verbitsky showed~\cite{Verbitsky} that the value of the $n$-fold repeated game approaches $0$,
however his argument uses the density Hales-Jewett theorem and hence gives a weak rate of decay (inverse Ackermann type bounds in $n$). More recently, researchers
have been trying to investigate multi-player games more systematically and managed to prove an exponential decay for a certain class of games known as expanding
games~\cite{DHVY}. This work also identified the {\sf GHZ} game as a bottleneck for current technique, saying that, in a sense, the {\sf GHZ} game exhibits the worst possible
correlations between questions for which existing information-theoretic techniques are incapable of handling.

A sequence of recent works~\cite{HR} (subsequently simplified by~\cite{GHMRZ}) managed to prove stronger parallel repetition theorems for the {\sf GHZ} game,
and indeed as suggested by~\cite{DHVY} this development led to a parallel repetition theorem for a certain class of $3$-player games~\cite{GHMRZ2,GMRZ},
namely for the class of games with binary questions. Quantitatively, they showed that ${\sf val}({\sf GHZ}^{\otimes n})\leq 1/n^{\Omega(1)}$, and
subsequently that for any $3$-player game $G$ with ${\sf val}(G)<1$ whose questions are binary, one has that ${\sf val}(G^{\otimes n})\leq 1/n^{\Omega(1)}$.
The techniques utilized by these works is a combination of information theoretic techniques (as used in the case of $2$-player games) and Fourier analytic techniques.

\subsection{Our Result}
The main result of this paper is an improved upper bound for the value of the $n$-fold repeated {\sf GHZ} game, which is exponential in $n$.
More precisely:
\begin{thm}\label{thm:main}
  There is $\eps>0$ such that for all $n$, ${\sf val}({\sf GHZ}^{\otimes n})\leq 2^{-\eps\cdot n}$.
\end{thm}
Such bounds cannot be achieved by the methods of \cite{HR,GHMRZ,GHMRZ2,GMRZ}, and we hope that the observations made herein would be useful towards
getting better parallel repetition theorems for more general classes of $3$-player games.

\subsection{Proof Idea}
Our proof of Theorem~\ref{thm:main} follows by reducing it to approximate sub-group type questions from additive combinatorics,
and our argument uses results of Gowers~\cite{Gowers}. Similar ideas have been
also explored in the TCS community (for example, by Samorodnitsky~\cite{Samorodnitsky}).
\skipi
Suppose $f\colon \{0,1\}^n \to\{0,1\}^n$, $g\colon \{0,1\}^n \to\{0,1\}^n$ and $h\colon \{0,1\}^n\to\{0,1\}^n$ represent the strategies of
Alice, Bob and Charlie respectively, and denote their success probability by $\eta$. Thus, we have that
\begin{equation}\label{eq:into}
\Prob{(x,y,z)\in S^n}{ f(x) \oplus g(y) \oplus h(z) = x\lor y \lor z}\geq \eta,
\end{equation}
where the operations are coordinate-wise. Using Cauchy-Schwarz it follows that if we sample $x,y,z$ and $u,v,w$ conditioned on $x\lor y\lor z = u\lor v\lor w$,
then $f(x) \oplus g(y) \oplus h(z) = f(u) \oplus g(v) \oplus h(w)$ with probability at least $\eta^2$, hence
$f(x) \oplus f(u) \oplus g(y) \oplus g(v) \oplus h(z) \oplus h(w) = 0$. What functions $f,g,h$ can satisfy this?
We draw an intuition from~\cite{BKM22}, that suggested that such advantage can only be gained from \emph{linear embeddings}. In this respect,
we are looking at the predicate $P\colon \Sigma^3\to \{0,1\}$ with alphabet $\Sigma = \{0,1\}^2$
defined as $P((x,u), (y,v), (z,w)) = 1$ if $x\lor y\lor z = u\lor v\lor w$, $x+y+z = 0$ and $u+v+w = 0$. A linear embedding is an Abelian group $(A,+)$
and a collection of maps $\phi\colon \Sigma \to A$, $\gamma\colon \Sigma\to A$ and $\delta \colon\Sigma\to A$ not all constant such that $\phi(x,u) + \gamma(y,v) + \delta(z,w) = 0$.
There are $2$ trivial linear embeddings into $(\mathbb{Z}_2,+)$: the projection onto the first coordinate as well as the projection onto the second coordinate.
Thus, one is tempted to guess that in the above scenario, the functions $f,g$ and $h$ must use these linear embeddings and thus be correlated with linear functions
over $\mathbb{Z}_2$.
Alas, it turns out that there is yet, another embedding which is less obvious: taking $(A,+) = (\mathbb{Z}_4,+)$, $\phi(x,u) = x+u$, $\gamma(y,v) = y+v$
and $\delta(z,w) = z+w$. This motivates us to look at the original problem and see if we can already see $(\mathbb{Z}_4,+)$ structure there.

\paragraph{Approximate Homomorphisms.}
For $(x,y,z)\in S$, if $x\lor y\lor z = 1$, then exactly
two of the variables are $1$; if $x\lor y\lor z = 0$, then all of $x,y,z$ are $0$. Thus, one can see that the check we are making is equivalent to checking that
$2f(x) + 2g(y) + 2h(z) = x+y+z\pmod{4}$. Indeed, on a given coordinate $i$, if $(x_i\lor y_i\lor z_i)$ is $1$, then $x_i+y_i+z_i = 2$ and the answers need to
satisfy that $f(x)_i + g(y)_i + h(z)_i = 1\pmod{2}$ which implies $2f(x)_i + 2g(y)_i + 2h(z)_i = 2\pmod{4}$. Similarly, if $(x_i\lor y_i\lor z_i) = 0$ then
$x_i + y_i + z_i = 0$ and the constraint says that we want $f(x)_i + g(y)_i + h(z)_i = 0\pmod{2}$ which implies that $2f(x)_i + 2g(y)_i + 2h(z)_i = 0\pmod{4}$.
Thus, the {\sf GHZ} test can be thought of as a system of equations modulo $4$, as suggested by the above intuition.
More precisely, defining $F\colon \{0,1\}^n\to\mathbb{Z}_4^n$ by $F(x)_i = 2f(x)_i - x_i$
and similarly $G, H\colon \{0,1\}^n\to\mathbb{Z}_4^n$ by $G(y)_i = 2g(y)_i - y_i$
and $H(z)_i = 2h(z)_i - z_i$, we have the following lemma:
\begin{lemma}\label{lem:reduce_to_addive}
  For each $x,y,z\in S^n$, $F(x) + G(y) + H(z) = 0\pmod{4}$ if and only if
  $f(x)_i \oplus g(y)_i \oplus h(z)_i = x_i\lor y_i \lor z_i$ for all $i=1,\ldots,n$.
  Consequently,
  \[
  \Prob{(x,y,z)\in S^n}{ F(x) + G(y) + H(z) = 0\pmod{4}}\geq \eta.
  \]
\end{lemma}
\begin{proof}
  Without loss of generality we focus on the first coordinate. If $(x_1,y_1,z_1) = (0,0,0)$, then by~\eqref{eq:into} we get that $f(x)_1 \oplus g(y)_1 \oplus h(z)_1 = 0$, hence
  either all of them are $0$ or exactly two of them are $1$, and in any case $2f(x)_1 + 2g(y)_1 + 2h(z)_1 = 0\pmod{4}$.
  Otherwise, without loss of generality $(x_1,y_1,z_1) = (1,1,0)$, and then by~\eqref{eq:into} we get $f(x)_1 \oplus g(y)_1 \oplus h(z)_1 = 1$,
  and there are two cases. If $f(x)_1 = g(y)_1 = h(z)_1 = 1$, then we get that
  $F(x)_1 + G(y)_1 + H(z)_1 = 2-1+2-1+2+0 = 0\pmod{4}$. Else, exactly one of them is $1$, say $f(x)_1 = 1$ and $g(y)_1 = h(z)_1 = 0$, and
  then $F(x)_1 + G(y)_1 + H(z)_1 = 2-1+0-1+0-0 = 0$.
\end{proof}
In words, Lemma~\ref{lem:reduce_to_addive} says that $F, G, H$ form an approximate ``cross homomorphism'' from $\mathbb{Z}_2^n$ to
$\mathbb{Z}_4^n$. Once we have made this observation, the proof is concluded by a routine application of powerful tools from additive combinatorics.

More specifically, we appeal to results of Gowers and show for any $F$ that satisfies Lemma~\ref{lem:reduce_to_addive} (for some $G$ and $H$)
must exhibit some weak linear behaviour. Specifically, we show that for such $F$ there is a shift $s\in\mathbb{Z}_4^n$ such that $F(x) \in s+\{0,2\}^n$
for at least $\eta'=\Omega(\eta^{10^4})$ fraction of inputs. On the other hand, on such points $x$ we get that $2f(x) - x = F(x) = s+L(x)$ for some $L(x)\in \{0,2\}^n$,
and noting that this must hold modulo $2$ we get that there can only be one such point, $x = -s\pmod{2}$. Thus, $\eta' \leq 2^{-n}$, giving an exponential bound
on $\eta$.

\section{Proof of Theorem~\ref{thm:main}}
\subsection{From Testing to Additive Quadruples}
We need the following definition:
\begin{definition}
  Let $(A,+), (B,+)$ be Abelian groups, and let $F\colon A^n\to B^n$. We say $(x,y,u,v)\in A^n\times A^n\times A^n\times A^n$
  is an additive quadruple if $x+y = u+v$ and $F(x) + F(y) = F(u) + F(v)$.
\end{definition}
In our application, we will always have $A = \{0,1\}$. For convenience we denote $N = 2^{n}$. Thus, it is clear that the number of additive quadruples
is always at most $N^{3}$ (as this is the number of solutions to $x+y = u+v$). The following lemma asserts that if
$F, G, H\colon \{0,1\}^n\to B^n$ are functions such that $F(x) + G(y) + H(z) = 0$ for at least $\eta$ of the triples $x,y,z$ satisfying
$x\oplus y = z$ (such as the one given in Lemma~\ref{lem:reduce_to_addive}), then each one of the functions $F, G$ and $H$ has a substaintial
amount of additive quadruples.
\begin{lemma}\label{lem:test_to_additive}
  Suppose that $F,G,H\colon \{0,1\}^n\to B^n$ satisfy that
  \[
  \Prob{(x,y,z)\in S^n}{F(x)+G(y)+H(z) = 0}\geq \eta.
  \]
  Then $F$ has at least $\eta^4 N^3$ additive quadruples.
\end{lemma}
\begin{proof}
  By the premise and Cauchy-Schwarz
  \begin{align*}
  \eta^2
  = \Expect{y}{\Expect{x}{1_{G(y) = -F(x)-H(x\oplus y)}}}^2
  &\leq
  \Expect{y}{\Expect{x}{1_{G(y) = -F(x)-H(x\oplus y)}}^2}\\
  &= \Expect{y}{\Expect{x,x'}{1_{G(y) = -F(x)-H(x\oplus y)}1_{G(y) = -F(x')-H(x'\oplus y)}}}\\
  &\leq \Expect{x,x',y}{1_{F(x)-F(x') = H(x'\oplus y) - H(x\oplus y)}}.
  \end{align*}
  Making change of variables, we get that $ \eta^2\leq  \Expect{x,u,u'}{1_{F(x)-F(x\oplus u\oplus u') = H(u') - H(u)}}$. Squaring and using Cauchy-Schwarz again
  we get that
  \begin{align*}
  \eta^4\leq
  \Expect{x,u,u'}{1_{F(x)-F(x\oplus u\oplus u') = H(u') - H(u)}}^2
  &\leq \Expect{u,u'}{\Expect{x}{1_{F(x)-F(x\oplus u\oplus u') = H(u') - H(u)}}^2}\\
  &\leq \Expect{u,u'}{\Expect{x,x'}{1_{F(x)-F(x\oplus u\oplus u') = F(x') - F(x'\oplus u\oplus u')}}},
  \end{align*}
  which by another change of variables is equal to $\Expect{x,y,u,v: x+y = u+v}{1_{F(x) + F(y) = F(u) + F(v)}}$,
  and the claim is proved.
\end{proof}

\subsection{From Additive Quadruples to Linear Structure}
We intend to use Lemma~\ref{lem:test_to_additive} to conclude a structural result for $F$, and
towards this end we show that a function that has many additive quadruples must exhibit some linear structure.
The content of this section is a straight-forward combination of well-known results in
additive combinatorics, and we include it here for the sake of completeness.
We need the notions of Freiman homomorphism, sum-sets and a result of Gowers~\cite{Gowers}. We begin with two definitions:
\begin{definition}
  Let $(A,+)$ and $(B,+)$ be Abelian groups, let $n\in\mathbb{N}$ and let $\mathcal{A}\subseteq A^n$.
  A function $\phi\colon \mathcal{A}\to B^n$ is called a Freiman homorphism of order $k$ if for
  all $a_1,\ldots,a_k\in\mathcal{A}$ and $b_1,\ldots,b_k\in \mathcal{A}$ such that
  $a_1+\ldots+a_k = b_1+\ldots+b_k$ it holds that
  \[
        \phi(a_1)+\ldots+\phi(a_k) = \phi(b_1)+\ldots+\phi(b_k).
  \]
\end{definition}
\begin{definition}
  Let $(A,+)$ be an Abelian group, let $n\in\mathbb{N}$ and let $\mathcal{A},\mathcal{B}\subseteq A^n$. We define
  \[
  \mathcal{A} + \mathcal{B} = \sett{a+b}{a\in \mathcal{A}, b\in \mathcal{B}}.
  \]
  If $\mathcal{A} = \mathcal{B}$, we denote the sum-set $\mathcal{A} + \mathcal{B}$ more succinctly as $2\mathcal{A}$, and more generally
  $k\mathcal{A}$ denotes the $k$-fold sum set of $\mathcal{A}$.
\end{definition}

We need a result of Gowers~\cite{Gowers} asserting that a function $F$ with many additive quadruples can be restricted to a relatively large set and yield
a Freiman homomorphism. Gowers states and proves the statement for $\mathbb{Z}_N$, and we adapt his proof for our setting.
For the proof we need two notable results in additive combinatorics.
The first of which is the Balog-Szemer\'edi-Gowers theorem, and we use the version from~\cite{Schoen}:
\begin{thm}[Balog-Szemer\'edi-Gowers]\label{thm:BSG}
  Let $G$ be an Abelian group, and suppose that $\Gamma\subseteq G$ contains at least $\xi \card{\Gamma}^3$ additive quadruples, that is,
  $\card{\sett{(x,y,z,w)\in \Gamma^4}{x+y = z+w}}\geq \xi\card{\Gamma}^3$. Then there exists $\Gamma'\subseteq \Gamma$ of size at least
  $\Omega(\xi\card{\Gamma})$ such that $\card{\Gamma'-\Gamma'}\leq O(\xi^{-4}\card{\Gamma'})$.
\end{thm}
The second result we need is Pl\"{u}nnecke's inequality~\cite{Plu,RuszaPlu} (see also~\cite{petridis2012new}):
\begin{thm}[Pl\"{u}nnecke's inequality]\label{thm:plu}
  Let $G$ be an Abelian group, and suppose that $\Gamma\subseteq G$ has $\card{\Gamma-\Gamma}\leq C\card{\Gamma}$. Then
  $\card{m\Gamma - r\Gamma}\leq C^{m+r}\card{\Gamma}$.
\end{thm}
\begin{lemma}[Corollary 7.6 in~\cite{Gowers}]\label{lem:gowers}
  Let $n\in\mathbb{N}$, and suppose that a function $\phi\colon \mathbb{Z}_2^n\to \mathbb{Z}_4^n$ has at least $\xi\card{\mathbb{Z}_2^{n}}^3$ additive quadruples.
  Then there exists $\mathcal{A}\subseteq \mathbb{Z}_2^n$ such that $\phi|_{\mathcal{A}}$ is a Freiman homomorphism
  of order $8$ and $\card{\mathcal{A}}\geq \Omega(\xi^{257}\card{\mathbb{Z}_2^n})$.
\end{lemma}
\begin{proof}
  Let $\Gamma = \sett{(x,\phi(x))}{x\in\mathbb{Z}_2^n}$ be the graph of $\phi$, and think of it as a set in the Abelian group $\mathbb{Z}_2^n\times \mathbb{Z}_4^n$.
  Then $\Gamma$ contains at least $\xi\card{\mathbb{Z}_2^{n}}^3=\xi\card{\Gamma}^3$ solutions to $\gamma_1 + \gamma_2 = \gamma_3 + \gamma_4$, hence
  by Theorem~\ref{thm:BSG} we may find $\Gamma'\subseteq \Gamma$ such that $\card{\Gamma'}\geq \Omega(\xi\card{\Gamma})$ and
  $\card{\Gamma' - \Gamma'}\leq O(\xi^{-4} \card{\Gamma'})$. By Theorem~\ref{thm:plu} we get that
  $\card{16\Gamma' - 16\Gamma'}\leq O(\xi^{-32\cdot 4} \card{\Gamma'}) \leq C\cdot \card{\Gamma'}$ where $C = O(\xi^{-128})$.

  Let $\mathcal{Y} = \sett{y\in\mathbb{Z}_4^n}{(0,y)\in 8\Gamma' - 8\Gamma'}$; we claim that $\card{\mathcal{Y}}\leq C$ and towards contradiction
  we assume the contrary.
  First, note that we may choose $\card{\Gamma'}$ distinct values of $x$ such that $(x,w_x)\in 8\Gamma' - 8\Gamma'$ for some $w_x$.
  Indeed, we can fix any $15$ elements $(x_i,w_i)\in \Gamma'$ for $i=1,\ldots,15$, and range over all $\card{\Gamma'}$ pairs $(x,w_x)\in \Gamma'$
  to get $\card{\Gamma'}$ elements $(x+x'-x'', w_x+w'-w'')\in 8\Gamma' - 8\Gamma'$ where $x' = x_1+\ldots+x_7$, $x'' = x_8+\ldots+x_{15}$
  and $w' = w_1+\ldots+w_7$ and $w'' = w_8+\ldots+w_{15}$, which have distinct first coordinate. Thus, looking at the $\card{\Gamma'}$ elements
  $(x,w_x)\in 8\Gamma' - 8\Gamma'$ with distinct first coordinate, we get that $(x,w_x+y)\in 16\Gamma' - 16\Gamma'$ for all $x$ and $y\in\mathcal{Y}$,
  hence $\card{16\Gamma' - 16\Gamma'}> C \card{\Gamma'}$, in contradiction.
  The set $\mathcal{Y}$ will be useful for us as for any $x\in \mathbb{Z}_2^n$, we may define $\mathcal{Y}_x = \sett{y}{(x,y)\in 4\Gamma' - 4\Gamma'}$
  and get that $\mathcal{Y}_x - \mathcal{Y}_x \subseteq \mathcal{Y}$.

  Take $t = \log(C)+1$, choose $I_1,\ldots,I_t\subseteq [n]$ independently and uniformly and consider
  \[
  \mathcal{W} = \sett{y\in\mathbb{Z}_4^n}{\sum\limits_{j\in I_i} y_j = 0~\forall i=1,\ldots,t}.
  \]
  We note that the $0$ vector is always in $\mathcal{W}$, but any other $y\in\mathbb{Z}_4^n$ is in $\mathcal{W}$ with probability at most $2^{-t}$.
  Indeed, if $y$'s entries are all $\{0,2\}$-valued then $y$ can be in $\mathcal{W}$ only if $y/2$ satisfies $t$ randomly chosen equations modulo $2$,
  which happens with probability $2^{-t}$. If there are entries of $y$ that are either $1$ or $3$, then we get that $y\pmod{2}$ is a non-zero vector
  that must satisfy $t$ randomly chosen equations modulo $2$, which happens with probability $2^{-t}$. Thus,
  $\Expect{}{\card{\mathcal{Y}\cap \mathcal{W}\setminus\set{0}}}\leq 2^{-t} \card{\mathcal{Y}}<1$, so we may choose $\mathcal{W}$
  such that $\mathcal{Y}\cap \mathcal{W} = \set{0}$.

  For an $a\in\mathbb{Z}_4^n$ we define $\Gamma'_a = \sett{(x,y)\in\Gamma'}{y\in a+\mathcal{W}}$. We claim that there is a choice for $a$ such that
  (1) $\card{\Gamma_a'}\geq 4^{-t} \card{\Gamma'}\geq \Omega(\xi^{257} \card{\mathbb{Z}_2^n})$, and
  (2) taking $\mathcal{A} = \sett{x}{\exists y\text{ such that }(x,y)\in \Gamma_a'}$, the function $\phi|_{\mathcal{A}}$ is a Freiman homomorphism of
  order $8$.
  Together, this gives the statement of the lemma.

  For the first item we have
  \[
  \Expect{a}{\card{\Gamma_a'}} = \sum\limits_{(x,y)\in \Gamma'} \Prob{a}{y\in a+\mathcal{W}}
  = \sum\limits_{(x,y)\in \Gamma'} \Prob{a}{y-a\in \mathcal{W}}
  \geq \sum\limits_{(x,y)\in \Gamma'} 4^{-t}
  =4^{-t}\card{\Gamma'},
  \]
  so there is an $a$ such that $\card{\Gamma_a'}\geq 4^{-t}\card{\Gamma'}$, and we show that the second item holds for all $a$.

  Suppose towards contradiction that $\phi|_{\mathcal{A}}$ is not a Freiman homomorphism of order $8$.
  Thus we can find $x_1,\ldots,x_8\in \mathcal{A}$ and $x_1',\ldots,x_8'\in \mathcal{A}$ that have the same sum yet $\phi(x_1)+\ldots + \phi(x_8) \neq \phi(x_1')+\ldots + \phi(x_8')$. Denoting $x = x_1+\ldots+x_4 - x_5' - \ldots - x_8' = x_1'+\ldots+x_4' - x_5 - \ldots - x_8$,
  $y = \phi(x_1)+\ldots + \phi(x_4)-\phi(x_5')-\ldots-\phi(x_8')$ and $y' = \phi(x_1')+\ldots + \phi(x_4')-\phi(x_5)-\ldots-\phi(x_8)$
  so that $y\neq y'$, we get that
  $(x,y),(x,y')\in 4\Gamma_a' - 4\Gamma_a'\subseteq 4\Gamma' - 4\Gamma'$, so $y,y'\in \mathcal{Y}_x$. In particular, $y-y'\in \mathcal{Y}_x-\mathcal{Y}_x\subseteq \mathcal{Y}$.
  On the other hand, by choice of $\mathcal{A}$ we get that $\phi(x_i),\phi(x_i')\in a+\mathcal{W}$ for all $i$ and so $y,y'\in 4\mathcal{W}-4\mathcal{W} = \mathcal{W}$ and so
  $y-y'\in\mathcal{W}$.  It follows that $y-y'\in \mathcal{Y}\cap \mathcal{W}$,
  but by the choice of $\mathcal{W}$ this last intersection only contains the $0$ vector, and contradiction.
\end{proof}
Thus, combining Lemmas~\ref{lem:test_to_additive} and~\ref{lem:gowers} we are able to conclude that $F$ is a Freiman homomorphism of order $8$ when restricted
to a set $\mathcal{A}\subseteq \mathbb{Z}_2^n$ whose size is at least $\Omega(\eta^{1028} N)$. A Freiman homomorphism of order $8$ is also a Freiman
homomorphism of order $4$, and the following lemma shows this tells that there is a shift of $\{0,2\}^n$ in which $F(x)$ lies for many $x$'s:

\begin{lemma}\label{lem:frieman_prop}
  Let $\mathcal{A}\subseteq \mathbb{Z}_2^n$ and suppose that $\phi\colon \mathcal{A}\to\mathbb{Z}_4^n$ is a Freiman homomorphism of order $4$. Then
  there is $s\in \mathbb{Z}_4^n$ such that for all $x\in\mathcal{A}$, $\phi(x) \in s+\{0,2\}^n$.
\end{lemma}
\begin{proof}
  Choose any $a\in\mathcal{A}$ and let $s = \phi(a)$. Then for any $x\in\mathcal{A}$, applying the Freiman homomorphism condition on the tuples
  $(x,x,a,a)$ and $(a,a,a,a)$ that have the same sum over $\mathbb{Z}_2^n$, we get that $2\phi(x) + 2\phi(a) = 4\phi(a)=0$, so $2(\phi(x) - s) = 0$.
  This implies that $\phi(x) - s\in\{0,2\}^n$, and the proof is concluded.
\end{proof}

Combining the last two lemmas we get the following corollary.
\begin{corollary}\label{cor:final}
  Suppose that $F\colon \mathbb{Z}_2^n\to \mathbb{Z}_4^n$ is a function for which there are $G, H\colon \mathbb{Z}_2^n\to \mathbb{Z}_4^n$
  such that $\Prob{(x,y,z)\in S^n}{F(x) + G(y) + H(z) = 0}\geq \eta$. Then there is $s\in \mathbb{Z}_4^n$ such that
  \[
  \Prob{x\in \mathbb{Z}_2^n}{F(x) \in \{0,2\}^n + s}\geq \Omega(\eta^{1028}).
  \]
\end{corollary}
\begin{proof}
  By Lemma~\ref{lem:test_to_additive} we get that $F$ has at least $\eta^4 N^3$ additive quadruples, so by Lemma~\ref{lem:gowers} there
  is $\mathcal{A}\subseteq \mathbb{Z}_2^n$ of size at least $\Omega(\eta^{1028} N)$ such that $F|_{\mathcal{A}}$ is a Freiman homomorphism.
  Applying Lemma~\ref{lem:frieman_prop} we conclude that there is $s\in\mathbb{Z}_4^n$ such that $F(x)\in s+\{0,2\}^n$ for all $x\in\mathcal{A}$
  and the proof is concluded.
\end{proof}

\subsection{Concluding Theorem~\ref{thm:main}}
Let $f,g,h\colon\{0,1\}^n\to\{0,1\}^n$ be strategies that achieve value at least $\eta$ in ${\sf GHZ}^{\otimes n}$,
and define $F\colon \mathbb{Z}_2^n\to\mathbb{Z}_4^n$ by $F(x) = 2f(x) - x$ and similarly $G(y) = 2g(y) - y$ and $H(z) = 2h(z) - z$.
By Lemma~\ref{lem:reduce_to_addive} we get that $\Prob{(x,y,z)\in S^n}{F(x) + G(y) + H(z) = 0}\geq \eta$, hence by Corollary~\ref{cor:final}
there is $s\in\mathbb{Z}_4^n$ such that $\Prob{x\in \mathbb{Z}_2^n}{F(x) \in s+\{0,2\}^n}\geq \eta'$ for $\eta' = \Omega(\eta^{1028})$.
For any such $x$, we get that $2f(x) - x = F(x) = s + L(x)$ where $L(x)\in \{0,2\}^n$, and so $x = -s + 2f(x)-L(x)$. Note that
this is equality modulo $4$ hence it implies it also holds modulo $2$. We also have that $2f(x) - L(x)\in \{0,2\}^n$ so this vanishes
modulo $2$, hence we get that $x = -s\pmod{2}$. In other words, there can be at most single $x$ such that $F(x) \in s+\{0,2\}^n$
and so $\Prob{x\in \mathbb{Z}_2^n}{F(x) \in s+\{0,2\}^n}\leq 2^{-n}$. Combining, we get that $\eta'\leq 2^{-n}$ and so $\eta\leq 2^{-n/1028 + O(1)}$.
\bibliography{ref}
\bibliographystyle{plain}
\end{document}